\documentclass[aps,prx,twocolumn,superscriptaddress,nofootinbib]{revtex4-2}

\usepackage{amsmath,amssymb,amsthm}
\usepackage{graphicx}
\usepackage[hidelinks]{hyperref}

\newtheorem{proposition}{Proposition}
\newtheorem{remark}{Remark}
\newcommand{\kdot}{\dot{\kappa}}
\newcommand{\R}{\mathbb{R}}

\begin{document}

\title{Untrainable elements determine what physical
learning remembers}

\author{Bijaya Dangol}
\affiliation{Independent researcher}

\begin{abstract}
Physical learning rules such as equilibrium propagation (EP), coupled
learning (CL), and adjoint coupled learning (AL) train resistive networks
through local measurements. Trainable edges outnumber output constraints,
so training reaches a manifold of solutions and the learned function is
decided by where on that manifold it lands. Two properties of these
systems could decide it, and they have not been separated: the circuit's
invariance under rescaling every conductance, and the rule's conservation
of the mass $K = \tfrac12\sum_e \kappa_e^2$. We separate them. In a
circuit whose every element is trainable, all three vector fields are
homogeneous in the conductances, so the initialization scale is provably
inert. An element the rule does not adjust is not rescaled with the
others and breaks that homogeneity whatever its constitutive law. Across
twenty topologies the learned function then moves with the initialization
scale by a median of twelve percent with fixed rectifiers and eight
percent with fixed linear resistors, against $3\times10^{-8}$ when every
element is trainable, and a single fixed rectifier already produces the
whole effect. The conservation law is not what protects the function: AL,
which we prove dissipates the mass at exactly twice its own loss,
remembers its initialization as strongly as the rules that conserve it,
and the memory survives in runs where $K$ is conserved to $10^{-4}$.
Two dissociations make the separation sharp: raising the number of fixed
elements from one to eight multiplies the conservation drift by five
thousand and leaves the memory unchanged, while the all-trainable circuit
under AL drifts comparably and remembers nothing at all. What
the rule's conservation structure does control is the quality of the
solution reached. At matched training loss AL is worse than EP and CL
in four of six small circuits, by a median of three to seven percent, but
the ordering is not stable across training-loss checkpoints and does not
reproduce at fifty nodes, where the three rules are indistinguishable. Physical learning therefore carries two
independent inductive biases, one belonging to the circuit and one to the
rule, and only the first is a memory of how the device was built.
\end{abstract}

\maketitle

\section{Introduction}
\label{sec:intro}

A network of adjustable resistors can be trained to compute: contrastive
local rules compare two equilibrium states of the circuit, a free state
and a nudged state, and adjust each conductance using only quantities
measurable at that edge
\cite{stern2021supervised,stern2023learning}. Equilibrium
propagation (EP) nudges by applying a small force proportional to the
output error \cite{scellier2017equilibrium}; coupled learning (CL) clamps
the outputs a small step toward their targets \cite{stern2021supervised}.
Both rules have trained laboratory networks
\cite{dillavou2022demonstration,dillavou2024machine}, and both are
candidates for learning hardware whose energy cost is set by physics
rather than by digital arithmetic. They also inherit a question posed
broadly for adaptive matter: what a driven physical system retains of
its history. In driven many-body systems that retention is quantified
in the matter itself, where the information a system holds about its
drive is recoverable from its trajectories more fully than from work or
dissipation alone \cite{zhong2021machine}. For a trained circuit the
history is its initialization, and we ask what of it the learned
function keeps.

A convergence theory exists for these rules: in linear circuits EP and CL
converge locally to a solution under a coercivity condition
\cite{mcginnis2026coercivity},
and in the small-nudge limit both exactly conserve the conductance mass
$K = \tfrac12 \sum_e \kappa_e^2$ \cite{mcginnis2026conservation}. The same
analysis introduces a third rule, adjoint coupled learning (AL), which
restores the gradient-flow structure that CL lacks
\cite{mcginnis2026coercivity}. Convergence, however, is only half of
learning. Trainable edges outnumber output constraints, so the parameter
configurations that satisfy a training objective form a manifold, and a
convergence theorem says that training reaches this manifold, not where
it lands. The landing point is what determines the learned function on
inputs never seen in training. For gradient descent in deep networks that
selection question is the subject of the implicit-bias literature, where
the scale of the initialization controls which minimizer the trajectory
reaches \cite{chizat2019lazy,woodworth2020kernel}. For physical learning
it has been asked of power, where the initial conductances are chosen to
land on low-power members of the solution set
\cite{dillavou2022demonstration,stern2024training}, and of the
microscopic structure a trained network keeps of its task
\cite{stern2025physical,guzman2025microscopic}. It has not been asked of
the conservation law, or of the learning rule itself as a source of
inductive bias.

Two properties of these systems are candidates for deciding the landing
point, and they are easy to confuse because both concern the
conductances. The first belongs to the circuit: rescaling every
conductance leaves the input-output map unchanged, so the scale of the
initialization is invisible to the task. The second belongs to the rule:
EP and CL conserve $K$ exactly along the continuous-time flow, and the
conservation survives the removal of edges whose conductance reaches zero
\cite{mcginnis2026conservation}, so the initialization norm is never
forgotten. Gradient flow in deep networks has conserved quantities of the
second kind too, tied to architectural symmetries
\cite{kunin2021neural,marcotte2023abide}, but practical training breaks
them through discrete steps, weight decay and normalization, and what
selection survives is a property of the trajectory rather than of an
exact invariant. Physical learning sits at the opposite pole, which makes
it the setting in which the two candidates can be told apart.

Telling them apart takes an experiment and three preliminary facts. At a
single output the comparison is empty, because the EP and CL flows have
identical orbits up to a time reparametrization
(Remark~\ref{rem:equiv}), which is why no rule-dependent bias has been
reported and why any test must start at two outputs. The third rule
supplies the probe: AL dissipates $K$ at exactly twice its own loss
(Proposition~\ref{prop:dissipation}), so it has the circuit's invariance
without the rule's conservation. And that invariance is exact but
conditional: with every element trainable, the scaling map
$\kappa \mapsto c\kappa$ sends learning orbits to learning orbits for all
three rules (Proposition~\ref{prop:homogeneity}), while an element the
rule does not adjust breaks it. If the conservation law were what fixed
the learned function, AL would forget what EP and CL remember.

It does not. Over twenty topologies the learned function moves with the
initialization scale by a median of twelve percent with fixed rectifiers
and eight percent with fixed linear resistors, against $3\times10^{-8}$
when every element is trainable, and one fixed rectifier is enough
(Sec.~\ref{sec:memory}). AL's memory is the same as EP's to within the
ensemble spread, and the effect survives where $K$ is conserved to
$10^{-4}$. The memory is the circuit's, not the rule's. The rule's
conservation structure governs something else: at matched training loss
AL selects measurably worse solutions than EP and CL in four of six small
circuits, an effect we do not reproduce at fifty nodes
(Sec.~\ref{sec:generalization}).

\section{Setup}
\label{sec:setup}

We follow the framework and notation of
Refs.~\cite{mcginnis2026coercivity,mcginnis2026conservation}. A circuit is
a connected graph on $N$ nodes with $M$ trainable edges; edge $e$ carries
a conductance $\kappa_e > 0$ and an incidence vector $d_e \in \R^N$. Node
voltages $x \in \R^N$ minimize the co-content, which for a linear circuit
is the dissipated power
\begin{equation}
  G(x;\kappa) = \tfrac12 \langle x, L(\kappa)\, x\rangle,
  \qquad
  L(\kappa) = \sum_{e} \kappa_e\, d_e d_e^{\top}.
\end{equation}
Inputs are imposed as voltage constraints $P^\top x = v$ on a set of input
nodes that includes ground, and outputs are read at disjoint nodes through
$Q^\top x$ with $Q \in \R^{N\times O}$. The free state is
$x_0 = \arg\min \{G(x;\kappa) : P^\top x = v\}$, and for a target $w$ the
error is $r = Q^\top x_0 - w$ with loss $\Phi = \tfrac12 \|r\|^2$.

The three rules perturb the output nodes in different ways. EP adds a
force $\eta\, Qr$ to the energy. CL imposes the constraint
$Q^\top x_\eta = (1+\eta)\,Q^\top x_0 - \eta w$. AL replaces the free
state by the doubly constrained reference state
$x_0^* = \arg\min\{G : P^\top x = v,\ Q^\top x = w\}$, whose output
constraint force defines the adjoint error $\mu_0$, and nudges by
$\eta\mu_0$ \cite{mcginnis2026coercivity}. In each case the physically
measurable update is
$\Delta \kappa_e \propto \eta^{-1}
[\partial_{\kappa_e} G(x_\eta) - \partial_{\kappa_e} G(x_{\mathrm{ref}})]$.
We work throughout in the idealized continuous-time, zero-nudge limit,
where the three flows take the common form
$\kdot_e = (d_e^\top y)(d_e^\top x_{\mathrm{ref}})$ with
\begin{equation}
\label{eq:flows}
  y_{\mathrm{EP}} = R\,Q r,
  \quad
  y_{\mathrm{CL}} = R\,Q D r,
  \quad
  y_{\mathrm{AL}} = R_{\mathrm{AL}}\, Q D_{\mathrm{AL}}\, \mu_0 .
\end{equation}
Here $R$ is the input-constrained resolvent of
Ref.~\cite{mcginnis2026coercivity},
$D = (Q^\top R\, Q)^{-1}$ is a discrete Dirichlet-to-Neumann operator, and
the AL quantities are evaluated at the reference state. EP is exact
gradient flow of $\Phi$. CL is not the gradient flow of any loss; its
error is premultiplied by the state-dependent operator $D$. AL is exact
gradient flow of the adjoint loss $\Phi^* = \tfrac12\|\mu_0\|^2$
\cite{mcginnis2026coercivity}. Following
Ref.~\cite{mcginnis2026conservation}, an edge whose conductance reaches
zero is removed and stays removed.

Because $M > O$ in all cases of interest, the solution set
$\mathcal{S} = \{\kappa : Q^\top x_0(\kappa) = w\}$ is generically a
manifold of dimension $M - O$. The convergence results of
Ref.~\cite{mcginnis2026coercivity} guarantee, under a coercivity
condition and from initializations close enough to $\mathcal{S}$, that
training reaches $\mathcal{S}$; the same reference exhibits spurious
fixed points off $\mathcal{S}$, so the guarantee is local. Our subject
is where on
$\mathcal{S}$ each flow lands, and what the landing point computes on
held-out inputs. For a linear circuit the trained behavior is summarized
by the input-output map $A_C(\kappa) \in \R^{O\times I}$ defined by
$Q^\top x_0 = A_C(\kappa)\, v$. The map is $0$-homogeneous in $\kappa$,
and crossbar circuits realize exactly the row-stochastic maps
\cite{mcginnis2026conservation}.

\section{Conservation structure and scale invariance}
\label{sec:theory}

Selection differences between EP and CL have gone unnoticed because in
the single-output setting common to most demonstrations there are none
to see. That the operator $D$ reduces to a positive scalar at $O=1$ is
noted in Ref.~\cite{mcginnis2026coercivity}, where it shows that EP and
CL lose coercivity under identical conditions, and in
Ref.~\cite{mcginnis2026conservation}, where it rules out limit cycles
for one-dimensional output. The same positivity determines the orbits,
and with them the selection.

\begin{remark}
\label{rem:equiv}
Let $O=1$, with inputs and outputs disjoint and $G(\cdot;\kappa)$ strictly
convex on the input-constraint surface. Then
$\kdot_{\mathrm{CL}} = D(\kappa)\,\kdot_{\mathrm{EP}}$ with
$D(\kappa) = (Q^\top R(\kappa)\, Q)^{-1}$ a positive scalar: $R$ is
positive semidefinite with $\ker R = \operatorname{range}(P)$, and
$Q \notin \operatorname{range}(P)$ because inputs and outputs are
disjoint, so $Q^\top R\, Q > 0$ \cite{mcginnis2026coercivity}. The two
fields therefore differ pointwise by a positive continuous scalar, a time
change maps one flow onto the other, and the two have identical orbits
and identical $\omega$-limit sets. Any selection difference between EP
and CL requires $O \ge 2$. Since Eq.~\eqref{eq:flows} is derived for
general separable energies in Ref.~\cite{mcginnis2026coercivity}, this
holds beyond linear circuits: single-output benchmarks, the configuration
of most published demonstrations, cannot reveal an inductive-bias
difference between EP and CL.
\end{remark}

The conservation law of Ref.~\cite{mcginnis2026conservation} states that
EP and CL conserve $K = \sum_i f_i(\kappa_i)$ whenever the energy
separates as $G = \sum_i g_i(\kappa_i)\,\widetilde G_i(x)$, with
$f_i' = g_i/g_i'$; for conductance training,
$K = \tfrac12\sum_e \kappa_e^2$. AL was introduced after that analysis and
its conservation behavior was not addressed. It has the matching
dissipation law, with the rate fixed by the rule's own objective.

\begin{proposition}[Dissipation law for AL]
\label{prop:dissipation}
Under the separability hypothesis above, the zero-nudge AL flow satisfies
\begin{equation}
  \dot K \;=\; -\lVert \mu_0 \rVert^2 \;=\; -\,2\,\Phi^{*} .
\end{equation}
Along the continuous-time flow $K$ therefore decreases strictly whenever
the objective is not yet met.
\end{proposition}

\begin{proof}
Writing the AL flow as
$\dot\kappa_i = \langle y, \partial_{\kappa_i}\nabla_x G(x_0^*)\rangle$,
\begin{align*}
  \dot K = \sum_i f_i'(\kappa_i)\,\dot\kappa_i
  &= \Big\langle y, \sum_i f_i'(\kappa_i)\,
      \partial_{\kappa_i}\nabla_x G(x_0^*) \Big\rangle \\
  &= \langle y, \nabla_x G(x_0^*) \rangle,
\end{align*}
using the separability identity of Ref.~\cite{mcginnis2026conservation}
at the reference state. The reference-state stationarity conditions give
$\nabla_x G(x_0^*) = -P\lambda_0 - Q\mu_0$, and the linearized system
defining $y$ imposes $P^\top y = 0$ and $Q^\top y = \mu_0$
\cite{mcginnis2026coercivity}. Hence
$\dot K = -\langle Q^\top y, \mu_0 \rangle = -\lVert\mu_0\rVert^2$.
\end{proof}

\begin{remark}
Proposition~\ref{prop:dissipation} is the zero-nudge counterpart of the
finite-nudge drift $\dot K = -\eta\, G(y;\kappa)$ that
Ref.~\cite{mcginnis2026conservation} identifies for one-sided nudging. AL
therefore carries, even at zero nudge, the parameter-scale drift that
Ref.~\cite{mcginnis2026conservation} attributes to non-conserving
nudging. In hardware whose conductances have
bounded operating ranges this is a design consideration in itself. Below
it serves two further purposes: it makes AL the probe that separates the
circuit's invariance from the rule's conservation
(Sec.~\ref{sec:memory}), and it governs the quality of the solution AL
selects (Sec.~\ref{sec:generalization}). We use the exact rate
$-2\Phi^*$ as the step-control law for our AL integrator
(Appendix~\ref{app:numerics}).
\end{remark}

A conserved quantity restricts each trajectory to a sphere
$\|\kappa\| = \mathrm{const}$, and one might expect the initialization
norm to index the learned function. In linear circuits it cannot.

\begin{proposition}[Initialization scale is inert in linear circuits]
\label{prop:homogeneity}
In a linear resistive circuit all of whose elements are trainable
conductances, the vector fields of EP, CL, and AL are homogeneous in
$\kappa$ of degrees $-1$, $0$, and $+1$:
\begin{equation}
\begin{aligned}
  F_{\mathrm{EP}}(c\kappa) &= c^{-1} F_{\mathrm{EP}}(\kappa), \\
  F_{\mathrm{CL}}(c\kappa) &= F_{\mathrm{CL}}(\kappa), \\
  F_{\mathrm{AL}}(c\kappa) &= c\, F_{\mathrm{AL}}(\kappa),
\end{aligned}
\qquad c > 0 .
\end{equation}
For each rule the map $\kappa \mapsto c\kappa$ sends learning orbits to
learning orbits, so the $\omega$-limit from $c\kappa_0$ is $c$ times the
$\omega$-limit from $\kappa_0$, and since $A_C(c\kappa) = A_C(\kappa)$,
the selected input-output map is the same from both initializations. The
conserved mass is functionally inert in linear circuits, and
initialization affects selection only through its direction.
\end{proposition}

\begin{proof}
The free state $x_0$ is $0$-homogeneous in $\kappa$: the voltage
constraints do not involve $\kappa$, and rescaling $G$ by $c$ leaves its
constrained minimizer unchanged. The input-constraint multipliers are
$1$-homogeneous. For EP, $y$ solves $L y + P\mu = Qr$ with $P^\top y = 0$
and $r$ $0$-homogeneous, so $y \to c^{-1} y$ under the scaling. For CL,
$y$ is pinned by the $\kappa$-independent constraints $P^\top y = 0$ and
$Q^\top y = r$, so $y$ is $0$-homogeneous. For AL, $\mu_0$ is
$1$-homogeneous and $Q^\top y = \mu_0$ forces $y \to c\,y$. In each case
$\kdot_e = (d_e^\top y)(d_e^\top x_{\mathrm{ref}})$ inherits the degree of
$y$. A homogeneous vector field of any degree has a scale-invariant
direction field, $F(c\kappa)/\|F(c\kappa)\| = F(\kappa)/\|F(\kappa)\|$,
so scaling maps orbits to time-reparametrized orbits. Edge removal
commutes with the scaling, since zero crossings scale with $\kappa$. The
$0$-homogeneity of $A_C$ completes the argument.
\end{proof}

In linear circuits the conserved norm is therefore remembered but inert.
The proposition also identifies where scale can act: any fixed, untrained
circuit element breaks the joint homogeneity and, as noted in
Ref.~\cite{mcginnis2026conservation}, the conservation law itself.
Section~\ref{sec:memory} tests exactly this.

\section{Solution selection in linear circuits}
\label{sec:linear}

The propositions constrain selection but do not quantify it. We measure
it first where they are exact. Circuits are
connected random graphs with $N=10$ to $12$ nodes, $M=19$ to $27$
trainable edges, three input nodes including ground, and $O=2$ to $4$
outputs. Targets are constructed from a reference conductance vector, so
the solution manifold is nonempty and of dimension $M-O$. Generation
details and seeds are in the repository (Appendix~\ref{app:numerics}).
Each run integrates the three flows from a common initialization to their
$\omega$-limits and compares the selected maps $A_C(\kappa_\infty)$ by the
relative Frobenius distance $\delta(A,B)=\|A-B\|/\|A\|$. As a numerical
floor we re-integrate EP with a different initial step size; the floor is
$\delta \approx 3\times10^{-9}$ in the median.

\begin{table}[t]
\caption{\label{tab:linear}%
Selection in linear circuits at $O=2$ (36 convergent runs over three
topologies, twelve initializations each). Distances are relative Frobenius
distances between selected input-output maps from a common
initialization. The direction spread is the median distance to the mean
map over 40 random initialization directions on a fixed task (one
topology), shown for EP; CL and AL give 0.164 and 0.183.}
\begin{ruledtabular}
\begin{tabular}{lccc}
 & median & min & max \\
\hline
numerical floor                  & $3.1\times10^{-9}$ &      & $6.8\times10^{-9}$ \\
$\delta(A_{\mathrm{EP}},A_{\mathrm{CL}})$ & $1.8\times10^{-5}$ & $1.4\times10^{-9}$ & $5.8\times10^{-4}$ \\
$\delta(A_{\mathrm{EP}},A_{\mathrm{AL}})$ & $1.2\times10^{-3}$ & $8.1\times10^{-6}$ & $1.3\times10^{-2}$ \\
$\delta(A_{\mathrm{CL}},A_{\mathrm{AL}})$ & $1.3\times10^{-3}$ & $8.2\times10^{-6}$ & $1.3\times10^{-2}$ \\
\hline
direction spread (EP) & $1.6\times10^{-1}$ & \multicolumn{2}{c}{} \\
\end{tabular}
\end{ruledtabular}
\end{table}

The selection differences predicted by Eq.~\eqref{eq:flows} are real
(Table~\ref{tab:linear}). With two outputs, EP and CL land on measurably
different points of the solution manifold, with a median $\delta$ of
$1.8\times10^{-5}$, four orders of magnitude above the floor, and AL
lands roughly seventy times farther from either. Consistent with the
mechanism, which is $D$ deviating from a scalar, the EP-CL gap correlates
with the strength of the effective coupling between output nodes
(Spearman rank correlation $0.41$, $n=36$ pooled over $O=2,3,4$), though
it does not grow monotonically with the number of outputs. Real, however,
is not the same as large. Over random initialization directions at fixed
task, the selected maps spread by $\delta \approx 0.16$, which exceeds
every rule effect by two to four orders of magnitude: in linear circuits
the inductive bias of physical learning is set almost entirely by where
training starts, and the rule contributes a small systematic correction.
Nor does the correction matter for generalization here. Training on a
single input-output pair against a ground-truth circuit and evaluating
the full map, the three rules generalize identically within statistical
resolution, with median relative errors $0.0879$, $0.0879$, and $0.0881$
over 30 runs; a three-pair training set determines the $2\times3$ map
completely and all rules recover it, which confirms identifiability.
Linear circuits therefore provide a clean null: rule choice moves the
selected solution but not its quality.

\section{Initialization memory from untrainable elements}
\label{sec:memory}

Proposition~\ref{prop:homogeneity} localizes where initialization scale
can matter, and the condition it needs is that every element be
trainable. An element the rule does not adjust is not rescaled with the
others, so the joint homogeneity fails whatever that element's
constitutive law. Nonlinearity is therefore not required, and separating
the two demands three arms rather than the usual two. On each of twenty
random topologies at $N=8$ we designate four edges and run the same task
three ways: with those edges fixed rectifiers, ideal diodes with reverse
leakage $\epsilon=10^{-4}$; with the same edges fixed \emph{linear}
resistors of the same conductance, untrained but not nonlinear; and with
all sixteen edges linear and trainable. The three arms share the graph,
the ports, the designated subset, and the initial conductances of every
edge either arm trains. Laboratory contrastive networks contain both
kinds of fixed component \cite{dillavou2024machine}, so the comparison is
the physically relevant one.

With rectifiers the co-content remains convex and piecewise quadratic,
and the zero-nudge flows of Eq.~\eqref{eq:flows} carry over with the
Hessian of $G$ at the reference state in place of $L$
(Appendix~\ref{app:numerics}). The trained behavior is no longer a
matrix, so we compare functions by their outputs on a fixed $5\times5$
grid of test inputs spanning both polarities, which switches the
rectifiers.

\begin{table*}[t]
\caption{\label{tab:memory}%
Initialization-norm memory over twenty random $N=8$ topologies, sixteen
edges each, three inputs including ground and two outputs. Arm A makes
four designated edges fixed rectifiers, arm B makes the same four edges
fixed linear resistors of the same conductance, and arm C makes all
sixteen trainable. Entries are the relative distance on the $5\times5$
grid between the function learned from $c\,\kappa_0$ and from $\kappa_0$
under EP, as ensemble median with interquartile range and full range. CL
agrees with EP to two digits throughout. Cells rest on 18 to 20
topologies; 32 of 900 runs did not reach training loss $10^{-11}$ within
the integration budget and are excluded. The lower block gives the
relative change in $K$ over the same runs.}
\begin{ruledtabular}
\footnotesize
\begin{tabular}{lcccc}
 & $c=1/4$ & $c=1/2$ & $c=2$ & $c=4$ \\
\hline
\multicolumn{5}{l}{\textit{Memory}, relative distance from the $c=1$ function}\\
A, fixed rectifiers & $2.1\times10^{-1}$ & $1.0\times10^{-1}$ & $8.9\times10^{-2}$ & $1.5\times10^{-1}$ \\
\quad interquartile & $1.5$ to $2.4\times10^{-1}$ & $6.6$ to $12\times10^{-2}$ & $6.4$ to $11\times10^{-2}$ & $1.1$ to $2.4\times10^{-1}$ \\
\quad range & $7.6\times10^{-2}$ to $4.8\times10^{-1}$ & $4.0$ to $16\times10^{-2}$ & $2.8$ to $17\times10^{-2}$ & $4.9\times10^{-2}$ to $3.4\times10^{-1}$ \\
B, fixed linear resistors & $1.2\times10^{-1}$ & $6.5\times10^{-2}$ & $6.4\times10^{-2}$ & $9.6\times10^{-2}$ \\
\quad interquartile & $3.7$ to $16\times10^{-2}$ & $2.2$ to $9.9\times10^{-2}$ & $3.0$ to $9.9\times10^{-2}$ & $6.6$ to $18\times10^{-2}$ \\
\quad range & $6.5\times10^{-6}$ to $2.9\times10^{-1}$ & $6.6\times10^{-6}$ to $1.7\times10^{-1}$ & $7.3\times10^{-7}$ to $1.8\times10^{-1}$ & $8.5\times10^{-7}$ to $4.1\times10^{-1}$ \\
C, all edges trainable & $2.3\times10^{-7}$ & $3.5\times10^{-8}$ & $1.7\times10^{-8}$ & $2.1\times10^{-8}$ \\
\quad range & $0.7$ to $17\times10^{-7}$ & $0.03$ to $3.5\times10^{-7}$ & $0.008$ to $0.8\times10^{-7}$ & $0.02$ to $0.9\times10^{-7}$ \\
\hline
paired ratio B/A, median & $0.59$ & $0.70$ & $0.96$ & $1.09$ \\
topologies with A above B & $15/18$ & $12/18$ & $10/18$ & $8/18$ \\
\hline
\multicolumn{5}{l}{\textit{Conservation drift} $|\Delta K|/K_0$, median (maximum)}\\
A, fixed rectifiers & $8.2\times10^{-2}$ ($6.6$) & $3.7\times10^{-2}$ ($1.5$) & $7.3\times10^{-3}$ ($2.6\times10^{-1}$) & $5.3\times10^{-3}$ ($4.4\times10^{-1}$) \\
B, fixed linear resistors & $6.1\times10^{-1}$ ($6.6$) & $1.0\times10^{-1}$ ($1.5$) & $3.7\times10^{-2}$ ($2.7\times10^{-1}$) & $4.3\times10^{-2}$ ($3.6\times10^{-1}$) \\
C, all edges trainable & $9.1\times10^{-9}$ ($4.8\times10^{-8}$) & $8.3\times10^{-9}$ ($4.7\times10^{-8}$) & $9.3\times10^{-13}$ ($6.0\times10^{-11}$) & $5.3\times10^{-15}$ ($2.7\times10^{-13}$) \\
\end{tabular}
\end{ruledtabular}
\end{table*}

Rescaling the initialization, with the same direction, the same task, and
convergence to the same training loss, changes the learned function by a
median of twelve percent under fixed rectifiers and eight percent under
fixed linear resistors, against $3\times10^{-8}$ when every element is
trainable (Table~\ref{tab:memory}). Both arms sit five to seven orders of
magnitude above the control, and the paired ratio between them has median
$0.75$ across the ensemble. Fixedness produces the memory; nonlinearity
is not required.

What nonlinearity supplies is concentration (Fig.~\ref{fig:dose}).
Sweeping the number of fixed edges from one to eight leaves the rectified
arm flat, median memory $0.146$, $0.138$, $0.132$, $0.153$, so a single
non-adjustable rectifier already produces the whole effect. At one fixed
edge the linear arm is the smaller of the two, median $0.043$ with a
paired ratio of $0.62$, but it exceeds the rectified arm in six of
eighteen topologies and its dependence on the number of fixed edges is
not monotonic, so we read no dose relation into it. At four fixed edges
the rectified arm is also the steadier: its memory never falls below
$2.8\times10^{-2}$, while the linear arm drops to $10^{-6}$ in one
topology of twenty; at one or two fixed edges single topologies of the
dose sweep fall to $2\times10^{-3}$ and $10^{-6}$. For hardware
the operative statement is that one non-adjustable component suffices.

\begin{figure}[t]
\includegraphics[width=\linewidth]{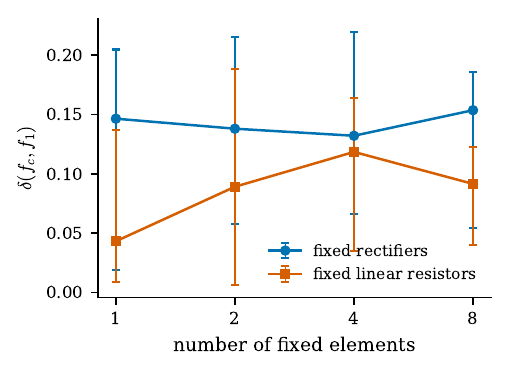}
\caption{\label{fig:dose}%
How much fixedness is needed. Median relative distance
$\delta(f_c,f_1)$ over ten topologies and the scales $c\in\{1/4,4\}$,
against the number of edges held fixed, with interquartile bars. The
rectified arm is flat: one fixed rectifier gives the same memory as
eight. The linear arm is smaller at a single fixed edge and does not
depend monotonically on the count.}
\end{figure}

The same term that breaks the homogeneity also breaks the conservation of
$K$, and the two effects are separable in the data. Drift with untrained
elements present is far larger than the linear case admits: median
$1.4\times10^{-2}$ under rectifiers and $5.3\times10^{-2}$ under fixed
linear resistors, with eight runs in each arm above unity. Yet the
memory does not follow it. In the sixteen runs whose drift stays below
$10^{-4}$ the memory is still $3.4$ to $26$ percent, median $11.4$
percent, and the fixed-linear arm drifts more than the rectified arm
while remembering less. Stratifying arm A at $|\Delta K|/K_0 = 10^{-3}$
moves the median memory from $10.6$ percent in the nineteen conserving
runs to $13.4$ percent in the fifty-four drifting ones, and pooling the
two fixed arms the medians are $10.6$ and $10.0$ percent. The
initialization scale is imprinted where the conservation law holds, not
only where it fails.

Two dissociations settle the point. Sweeping the number of fixed edges
from one to eight raises the median drift monotonically from
$7.2\times10^{-5}$ to $3.9\times10^{-1}$, a factor of five thousand,
while the median memory stays at $0.146$, $0.138$, $0.132$ and $0.153$:
four decades of change in how badly the conservation law fails, and no
movement in what the circuit remembers. In the all-trainable arm the
comparison runs the other way. There AL drifts by a median
$5.1\times10^{-3}$, a third of the rectified arm's $1.4\times10^{-2}$,
and its memory is exactly zero to machine precision in every run.
Substantial drift with no memory, and four decades of drift with constant
memory, on the same graphs in the same experiment.

In deep networks the initialization scale also selects the solution, but
through the trajectory the optimizer happens to take
\cite{chizat2019lazy,woodworth2020kernel}, and discrete steps, weight
decay and normalization erode that selection. Here the selection is a
property of the circuit rather than of the integrator. Ref.~\cite{stern2024training}
observes that a circuit's outputs depend on conductance ratios and not on
overall scale, which is Proposition~\ref{prop:homogeneity} read at fixed
$\kappa$; the memory measured here concerns which ratios training
reaches from a given starting scale, and the two statements are
compatible.

The dissociations vary how badly the conservation law fails; AL
discards it altogether. By Proposition~\ref{prop:dissipation} it burns
the mass at twice its own loss, so trajectories launched
from different initial spheres are drawn together as they descend. If the
conservation law were what fixed the learned function, AL would forget
what EP and CL are bound to keep. It does not. Over the ensemble the
median ratio of AL's memory to EP's is $1.00$, with AL smaller in seven
to nine of the seventeen or eighteen paired topologies at each scale and
larger in the rest. A
rule that destroys the invariant remembers its initialization exactly as
well as the rules that preserve it, which leaves the circuit's broken
homogeneity as the only mechanism the data support.

\begin{figure}[t]
\includegraphics[width=\linewidth]{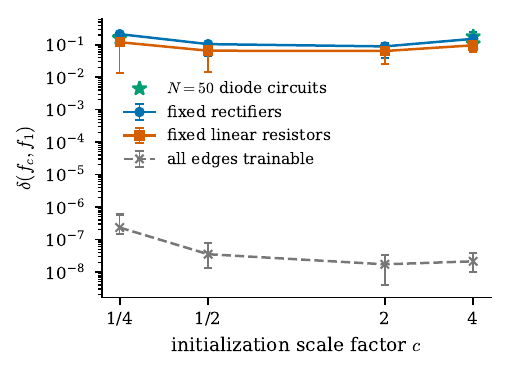}
\caption{\label{fig:memory}%
Initialization-norm memory across the three arms of
Table~\ref{tab:memory}: ensemble medians over twenty $N=8$ topologies
with interquartile bars, under EP. Fixed elements of either kind separate
the learned functions from the all-trainable case by five to seven orders
of magnitude, and the two fixed arms lie within a factor of two of each
other. Stars are two $N=50$ diode circuits on the same protocol, showing
that the effect is not an artifact of small circuits.}
\end{figure}

The memory effect is not an artifact of small circuits
(Fig.~\ref{fig:memory}). At $N=50$ nodes,
with 75 trainable edges and 25 rectifiers, the same protocol gives
function changes of 11 to 18 percent at $c=1/4$ and $c=4$. The all-linear
controls converge to zero as the integration budget grows, and their
finite-time deviations follow the degree ladder of
Proposition~\ref{prop:homogeneity}: the rule whose field slows under the
scaling is the one that lags, in the direction the degree predicts.

\section{Generalization at matched training loss}
\label{sec:generalization}

Do the rules differ in the quality of the solutions they select?
Comparisons of learned functions at the end of training are confounded by
convergence depth: the rules reach different training losses in equal
time, and test-error differences then reflect optimization speed rather
than selection. The confound is large enough to invert conclusions. At a
fixed integration budget AL, which trains fastest, appears to generalize
better than EP on the tasks below; at matched training loss the ordering
reverses. We therefore integrate each rule until its training loss first
crosses fixed checkpoints, $3\times10^{-6}$ and $3\times10^{-7}$, and
evaluate the learned functions at the crossings. Tasks are three-pair
training sets from ground-truth diode circuits, with six topologies and
six initializations each, giving 36 paired runs. Every run reached both
checkpoints.

\begin{table*}[t]
\caption{\label{tab:matched}%
Generalization at matched training loss in nonlinear circuits ($n=36$
paired runs over six topologies). Test error is the relative distance to
the ground-truth function on the input grid, reported as the pooled
median. ``AL worse'' counts paired runs in which AL's test error exceeds
EP's, and the penalty is the median of the paired relative difference
$(\mathrm{AL}-\mathrm{EP})/\mathrm{EP}$. The residual floor $\sqrt{2L}$
bounds the function resolution at checkpoint $L$.}
\begin{ruledtabular}
\begin{tabular}{lcccccc}
checkpoint & EP & CL & AL & AL worse & penalty & floor \\
\hline
$L=3\times10^{-6}$ & $0.0343$ & $0.0347$ & $0.0347$ & $29/36$ & $+6.7\%$ & $2.4\times10^{-3}$ \\
$L=3\times10^{-7}$ & $0.0215$ & $0.0217$ & $0.0233$ & $26/36$ & $+3.5\%$ & $7.7\times10^{-4}$ \\
\end{tabular}
\end{ruledtabular}
\end{table*}

At matched loss the dissipative rule loses on average
(Table~\ref{tab:matched}). AL's test error exceeds EP's in 29 of 36
paired runs at the first checkpoint (two-sided sign test,
$p = 3\times10^{-4}$) and 26 of 36 at the second ($p = 0.011$), with
median paired penalties of $6.7$ and $3.5$ percent. The effect is
topology-dependent: four of the six topologies show AL worse in
essentially every run, one shows no consistent ordering, and one shows
the reverse. The functions AL selects differ from EP's by
$\delta \approx 3$ to $6\times10^{-3}$, above the residual floors, so the
differences are selection rather than incomplete convergence. EP and CL
remain statistically indistinguishable: CL's error exceeds EP's in 21 and
18 of 36 runs, consistent with a fair coin, and their function difference
sits at the floor. AL meanwhile reaches any given training loss fastest
of the three. The faster training is a direct consequence of
the dissipation, since the same rate that burns the mass is the rate at
which AL descends its own objective. Whether the matched-loss penalty
shares that cause is tested below. The initialization
memory does not: it is the circuit's, and Sec.~\ref{sec:memory} shows
AL carries as much of it as EP does.

\begin{figure}[t]
\includegraphics[width=\linewidth]{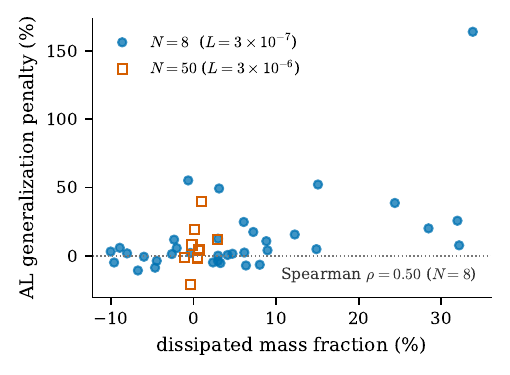}
\caption{\label{fig:dissipation}%
AL's generalization penalty against the fraction of conductance mass it
dissipated before reaching the training-loss checkpoint. Filled circles:
$N=8$, all 36 paired runs at $L=3\times10^{-7}$. Open squares: an
earlier, independent set of eight $N=50$ runs at $L=3\times10^{-6}$,
which reach the checkpoint after an order of magnitude less dissipation;
its median penalty of $+6$ percent is not resolved at $n=8$ and is not
reproduced by the fifteen-run set reported in the text. Thirteen of
the 36 runs sit at negative dissipated fraction: with fixed rectifiers
present the leak term of Sec.~\ref{sec:memory} adds to
$\dot K = -2\Phi^*$ and can outweigh it, so
Proposition~\ref{prop:dissipation}'s strict decrease is a statement about
the all-trainable flow.}
\end{figure}

Whether the dissipation causes the penalty is not established by these
data, and we report the tests rather than the reading
(Fig.~\ref{fig:dissipation}). Across the 36 runs at $N=8$ the penalty
co-varies with the fraction of mass AL dissipated before reaching the
deeper checkpoint, Spearman $0.50$, but not at the shallower one, where
the coefficient is $0.17$. The 36 runs are six topologies with six
initializations each and the runs within a topology move together, so
the level of independent replication is six: over the topology medians
the coefficient is $0.60$ with $n=6$. In the $N=50$ set the relation is
absent, Spearman $-0.011$ with $p=0.97$ over 15 runs. One supportive
correlation, one null and one absent relation do not identify a
mechanism, and the pattern is what a topology-level confound would also
produce: the reversed topology sits at low dissipation, as dissipation
would require, but one high-penalty topology dissipates little.

The penalty itself is likewise bounded to the small circuits. At $N=50$
the three rules are indistinguishable, with median test errors $0.0208$,
$0.0209$ and $0.0210$, AL worse in 6 of 15 paired runs and a median
paired penalty of $-3.3$ percent, so the point estimate crosses zero
rather than shrinking toward it. The checkpoint there is reached after
about a quarter of the dissipation of the $N=8$ runs, median fraction
$0.8$ against $3.3$ percent. We therefore record the matched-loss
penalty as a small-circuit observation that we cannot reproduce at
larger size, and claim no mechanism for it.

\section{Discussion}
\label{sec:discussion}

In Ref.~\cite{mcginnis2026conservation} the conserved mass plays a
stability role: it explains why training does not drift out of a device's
operating range, and symmetric clamping is prescribed to protect it. That
analysis concerns the scale of the parameters during training; ours
concerns the function the parameters select, and it separates two things
the conductances make easy to conflate. The circuit decides whether the
initialization scale is remembered at all. Any network holding a
component the rule does not adjust carries a functionally significant
imprint of the scale it was built at, one such component is enough, and
no choice of rule removes it: that burden belongs to the hardware, where
it is a calibration cost if uncontrolled and a design variable if
exploited. What the rule decides is narrower and less certain. It sets
the speed of training, where the dissipative rule is fastest by an order
of magnitude, and in the small circuits studied here it is associated
with a few percent of test error at matched training loss, an effect we
could not reproduce at fifty nodes. A designer therefore holds two levers
that act on different things, one of them established and one of them
provisional, and neither substitutes for the other. On the second lever
single-output benchmarks are silent, since they cannot distinguish EP
from CL at all.

Four limitations bound our claims. All experiments are on simulated
circuits of at most fifty nodes; the propositions hold at any size, and
the memory effect persists at the largest size we test, but hardware-scale
magnitudes are open. With untrained elements present the conservation of
$K$ is itself approximate, at the percent level in the median and worse
in a minority of runs, so the exactness of Proposition~\ref{prop:homogeneity}
is a statement about the all-trainable limit and not about the circuits
where the memory appears. The matched-loss penalty rests on six
topologies, is not universal among them, does not keep its direction
between training-loss checkpoints, and does not reproduce at fifty nodes;
we report it and do not build on it. No mechanism for it is established:
the dissipation-penalty association appears at one of three tests, and a
topology-level confound would produce the same pattern. A selection
theory for circuits with fixed components, where the interesting behavior
lives, is the natural next step.

In machine learning, conservation laws of training dynamics have been
catalogued but are broken by practical optimizers, and the initialization
scale acts through the trajectory rather than through an invariant
\cite{kunin2021neural,marcotte2023abide,chizat2019lazy,woodworth2020kernel}.
Physical learning realizes the opposite design point exactly, and it is
that exactness which lets the two mechanisms be told apart here rather
than only inferred. Conservation is native to the rule, scale invariance
is native to the circuit, and the memory of how a device was built
belongs to the second.

\section*{Data and code availability}

All results are reproducible from the repository at
\url{https://github.com/dangoldbj/physical-learning-memory}, which
contains the simulation package, the validation suite, and one script per
experiment with pinned seeds.

\bibliography{refs}

\begin{thebibliography}{15}%
\makeatletter
\providecommand \@ifxundefined [1]{%
 \@ifx{#1\undefined}
}%
\providecommand \@ifnum [1]{%
 \ifnum #1\expandafter \@firstoftwo
 \else \expandafter \@secondoftwo
 \fi
}%
\providecommand \@ifx [1]{%
 \ifx #1\expandafter \@firstoftwo
 \else \expandafter \@secondoftwo
 \fi
}%
\providecommand \natexlab [1]{#1}%
\providecommand \enquote  [1]{``#1''}%
\providecommand \bibnamefont  [1]{#1}%
\providecommand \bibfnamefont [1]{#1}%
\providecommand \citenamefont [1]{#1}%
\providecommand \href@noop [0]{\@secondoftwo}%
\providecommand \href [0]{\begingroup \@sanitize@url \@href}%
\providecommand \@href[1]{\@@startlink{#1}\@@href}%
\providecommand \@@href[1]{\endgroup#1\@@endlink}%
\providecommand \@sanitize@url [0]{\catcode `\\12\catcode `\$12\catcode
  `\&12\catcode `\#12\catcode `\^12\catcode `\_12\catcode `\%12\relax}%
\providecommand \@@startlink[1]{}%
\providecommand \@@endlink[0]{}%
\providecommand \url  [0]{\begingroup\@sanitize@url \@url }%
\providecommand \@url [1]{\endgroup\@href {#1}{\urlprefix }}%
\providecommand \urlprefix  [0]{URL }%
\providecommand \Eprint [0]{\href }%
\providecommand \doibase [0]{https://doi.org/}%
\providecommand \selectlanguage [0]{\@gobble}%
\providecommand \bibinfo  [0]{\@secondoftwo}%
\providecommand \bibfield  [0]{\@secondoftwo}%
\providecommand \translation [1]{[#1]}%
\providecommand \BibitemOpen [0]{}%
\providecommand \bibitemStop [0]{}%
\providecommand \bibitemNoStop [0]{.\EOS\space}%
\providecommand \EOS [0]{\spacefactor3000\relax}%
\providecommand \BibitemShut  [1]{\csname bibitem#1\endcsname}%
\let\auto@bib@innerbib\@empty
\bibitem [{\citenamefont {Stern}\ \emph {et~al.}(2021)\citenamefont {Stern},
  \citenamefont {Hexner}, \citenamefont {Rocks},\ and\ \citenamefont
  {Liu}}]{stern2021supervised}%
  \BibitemOpen
  \bibfield  {author} {\bibinfo {author} {\bibfnamefont {M.}~\bibnamefont
  {Stern}}, \bibinfo {author} {\bibfnamefont {D.}~\bibnamefont {Hexner}},
  \bibinfo {author} {\bibfnamefont {J.~W.}\ \bibnamefont {Rocks}},\ and\
  \bibinfo {author} {\bibfnamefont {A.~J.}\ \bibnamefont {Liu}},\ }\bibfield
  {title} {\bibinfo {title} {Supervised learning in physical networks: From
  machine learning to learning machines},\ }\href
  {https://doi.org/10.1103/PhysRevX.11.021045} {\bibfield  {journal} {\bibinfo
  {journal} {Physical Review X}\ }\textbf {\bibinfo {volume} {11}},\ \bibinfo
  {pages} {021045} (\bibinfo {year} {2021})}\BibitemShut {NoStop}%
\bibitem [{\citenamefont {Stern}\ and\ \citenamefont
  {Murugan}(2023)}]{stern2023learning}%
  \BibitemOpen
  \bibfield  {author} {\bibinfo {author} {\bibfnamefont {M.}~\bibnamefont
  {Stern}}\ and\ \bibinfo {author} {\bibfnamefont {A.}~\bibnamefont
  {Murugan}},\ }\bibfield  {title} {\bibinfo {title} {Learning without neurons
  in physical systems},\ }\href
  {https://doi.org/10.1146/annurev-conmatphys-040821-113439} {\bibfield
  {journal} {\bibinfo  {journal} {Annual Review of Condensed Matter Physics}\
  }\textbf {\bibinfo {volume} {14}},\ \bibinfo {pages} {417} (\bibinfo {year}
  {2023})}\BibitemShut {NoStop}%
\bibitem [{\citenamefont {Scellier}\ and\ \citenamefont
  {Bengio}(2017)}]{scellier2017equilibrium}%
  \BibitemOpen
  \bibfield  {author} {\bibinfo {author} {\bibfnamefont {B.}~\bibnamefont
  {Scellier}}\ and\ \bibinfo {author} {\bibfnamefont {Y.}~\bibnamefont
  {Bengio}},\ }\bibfield  {title} {\bibinfo {title} {Equilibrium propagation:
  Bridging the gap between energy-based models and backpropagation},\ }\href
  {https://doi.org/10.3389/fncom.2017.00024} {\bibfield  {journal} {\bibinfo
  {journal} {Frontiers in Computational Neuroscience}\ }\textbf {\bibinfo
  {volume} {11}},\ \bibinfo {pages} {24} (\bibinfo {year} {2017})}\BibitemShut
  {NoStop}%
\bibitem [{\citenamefont {Dillavou}\ \emph {et~al.}(2022)\citenamefont
  {Dillavou}, \citenamefont {Stern}, \citenamefont {Liu},\ and\ \citenamefont
  {Durian}}]{dillavou2022demonstration}%
  \BibitemOpen
  \bibfield  {author} {\bibinfo {author} {\bibfnamefont {S.}~\bibnamefont
  {Dillavou}}, \bibinfo {author} {\bibfnamefont {M.}~\bibnamefont {Stern}},
  \bibinfo {author} {\bibfnamefont {A.~J.}\ \bibnamefont {Liu}},\ and\ \bibinfo
  {author} {\bibfnamefont {D.~J.}\ \bibnamefont {Durian}},\ }\bibfield  {title}
  {\bibinfo {title} {Demonstration of decentralized physics-driven learning},\
  }\href {https://doi.org/10.1103/PhysRevApplied.18.014040} {\bibfield
  {journal} {\bibinfo  {journal} {Physical Review Applied}\ }\textbf {\bibinfo
  {volume} {18}},\ \bibinfo {pages} {014040} (\bibinfo {year}
  {2022})}\BibitemShut {NoStop}%
\bibitem [{\citenamefont {Dillavou}\ \emph {et~al.}(2024)\citenamefont
  {Dillavou}, \citenamefont {Beyer}, \citenamefont {Stern}, \citenamefont
  {Liu}, \citenamefont {Miskin},\ and\ \citenamefont
  {Durian}}]{dillavou2024machine}%
  \BibitemOpen
  \bibfield  {author} {\bibinfo {author} {\bibfnamefont {S.}~\bibnamefont
  {Dillavou}}, \bibinfo {author} {\bibfnamefont {B.~D.}\ \bibnamefont {Beyer}},
  \bibinfo {author} {\bibfnamefont {M.}~\bibnamefont {Stern}}, \bibinfo
  {author} {\bibfnamefont {A.~J.}\ \bibnamefont {Liu}}, \bibinfo {author}
  {\bibfnamefont {M.~Z.}\ \bibnamefont {Miskin}},\ and\ \bibinfo {author}
  {\bibfnamefont {D.~J.}\ \bibnamefont {Durian}},\ }\bibfield  {title}
  {\bibinfo {title} {Machine learning without a processor: Emergent learning in
  a nonlinear analog network},\ }\href
  {https://doi.org/10.1073/pnas.2319718121} {\bibfield  {journal} {\bibinfo
  {journal} {Proceedings of the National Academy of Sciences}\ }\textbf
  {\bibinfo {volume} {121}},\ \bibinfo {pages} {e2319718121} (\bibinfo {year}
  {2024})}\BibitemShut {NoStop}%
\bibitem [{\citenamefont {Zhong}\ \emph {et~al.}(2021)\citenamefont {Zhong},
  \citenamefont {Gold}, \citenamefont {Marzen}, \citenamefont {England},\ and\
  \citenamefont {Yunger~Halpern}}]{zhong2021machine}%
  \BibitemOpen
  \bibfield  {author} {\bibinfo {author} {\bibfnamefont {W.}~\bibnamefont
  {Zhong}}, \bibinfo {author} {\bibfnamefont {J.~M.}\ \bibnamefont {Gold}},
  \bibinfo {author} {\bibfnamefont {S.}~\bibnamefont {Marzen}}, \bibinfo
  {author} {\bibfnamefont {J.~L.}\ \bibnamefont {England}},\ and\ \bibinfo
  {author} {\bibfnamefont {N.}~\bibnamefont {Yunger~Halpern}},\ }\bibfield
  {title} {\bibinfo {title} {Machine learning outperforms thermodynamics in
  measuring how well a many-body system learns a drive},\ }\href
  {https://doi.org/10.1038/s41598-021-88311-7} {\bibfield  {journal} {\bibinfo
  {journal} {Scientific Reports}\ }\textbf {\bibinfo {volume} {11}},\ \bibinfo
  {pages} {9333} (\bibinfo {year} {2021})}\BibitemShut {NoStop}%
\bibitem [{\citenamefont {McGinnis}\ \emph
  {et~al.}(2026{\natexlab{a}})\citenamefont {McGinnis}, \citenamefont {Li},\
  and\ \citenamefont {Mori}}]{mcginnis2026coercivity}%
  \BibitemOpen
  \bibfield  {author} {\bibinfo {author} {\bibfnamefont {J.~A.}\ \bibnamefont
  {McGinnis}}, \bibinfo {author} {\bibfnamefont {X.}~\bibnamefont {Li}},\ and\
  \bibinfo {author} {\bibfnamefont {Y.}~\bibnamefont {Mori}},\ }\href@noop {}
  {\bibinfo {title} {Coercivity and local convergence of physical learning in
  linear circuits}} (\bibinfo {year} {2026}{\natexlab{a}}),\ \Eprint
  {https://arxiv.org/abs/2606.15443} {arXiv:2606.15443 [math.OC]} \BibitemShut
  {NoStop}%
\bibitem [{\citenamefont {McGinnis}\ \emph
  {et~al.}(2026{\natexlab{b}})\citenamefont {McGinnis}, \citenamefont {Kline},\
  and\ \citenamefont {Mori}}]{mcginnis2026conservation}%
  \BibitemOpen
  \bibfield  {author} {\bibinfo {author} {\bibfnamefont {J.~A.}\ \bibnamefont
  {McGinnis}}, \bibinfo {author} {\bibfnamefont {A.~G.}\ \bibnamefont
  {Kline}},\ and\ \bibinfo {author} {\bibfnamefont {Y.}~\bibnamefont {Mori}},\
  }\href@noop {} {\bibinfo {title} {A conservation law for equilibrium
  propagation and coupled learning}} (\bibinfo {year} {2026}{\natexlab{b}}),\
  \Eprint {https://arxiv.org/abs/2606.15444} {arXiv:2606.15444 [math.OC]}
  \BibitemShut {NoStop}%
\bibitem [{\citenamefont {Chizat}\ \emph {et~al.}(2019)\citenamefont {Chizat},
  \citenamefont {Oyallon},\ and\ \citenamefont {Bach}}]{chizat2019lazy}%
  \BibitemOpen
  \bibfield  {author} {\bibinfo {author} {\bibfnamefont {L.}~\bibnamefont
  {Chizat}}, \bibinfo {author} {\bibfnamefont {E.}~\bibnamefont {Oyallon}},\
  and\ \bibinfo {author} {\bibfnamefont {F.}~\bibnamefont {Bach}},\ }\bibfield
  {title} {\bibinfo {title} {On lazy training in differentiable programming},\
  }in\ \href@noop {} {\emph {\bibinfo {booktitle} {Advances in Neural
  Information Processing Systems (NeurIPS)}}}\ (\bibinfo {year}
  {2019})\BibitemShut {NoStop}%
\bibitem [{\citenamefont {Woodworth}\ \emph {et~al.}(2020)\citenamefont
  {Woodworth}, \citenamefont {Gunasekar}, \citenamefont {Lee}, \citenamefont
  {Moroshko}, \citenamefont {Savarese}, \citenamefont {Golan}, \citenamefont
  {Soudry},\ and\ \citenamefont {Srebro}}]{woodworth2020kernel}%
  \BibitemOpen
  \bibfield  {author} {\bibinfo {author} {\bibfnamefont {B.}~\bibnamefont
  {Woodworth}}, \bibinfo {author} {\bibfnamefont {S.}~\bibnamefont
  {Gunasekar}}, \bibinfo {author} {\bibfnamefont {J.~D.}\ \bibnamefont {Lee}},
  \bibinfo {author} {\bibfnamefont {E.}~\bibnamefont {Moroshko}}, \bibinfo
  {author} {\bibfnamefont {P.}~\bibnamefont {Savarese}}, \bibinfo {author}
  {\bibfnamefont {I.}~\bibnamefont {Golan}}, \bibinfo {author} {\bibfnamefont
  {D.}~\bibnamefont {Soudry}},\ and\ \bibinfo {author} {\bibfnamefont
  {N.}~\bibnamefont {Srebro}},\ }\bibfield  {title} {\bibinfo {title} {Kernel
  and rich regimes in overparametrized models},\ }in\ \href@noop {} {\emph
  {\bibinfo {booktitle} {Conference on Learning Theory (COLT)}}}\ (\bibinfo
  {year} {2020})\BibitemShut {NoStop}%
\bibitem [{\citenamefont {Stern}\ \emph {et~al.}(2024)\citenamefont {Stern},
  \citenamefont {Dillavou}, \citenamefont {Jayaraman}, \citenamefont {Durian},\
  and\ \citenamefont {Liu}}]{stern2024training}%
  \BibitemOpen
  \bibfield  {author} {\bibinfo {author} {\bibfnamefont {M.}~\bibnamefont
  {Stern}}, \bibinfo {author} {\bibfnamefont {S.}~\bibnamefont {Dillavou}},
  \bibinfo {author} {\bibfnamefont {D.}~\bibnamefont {Jayaraman}}, \bibinfo
  {author} {\bibfnamefont {D.~J.}\ \bibnamefont {Durian}},\ and\ \bibinfo
  {author} {\bibfnamefont {A.~J.}\ \bibnamefont {Liu}},\ }\bibfield  {title}
  {\bibinfo {title} {Training self-learning circuits for power-efficient
  solutions},\ }\href {https://doi.org/10.1063/5.0181382} {\bibfield  {journal}
  {\bibinfo  {journal} {APL Machine Learning}\ }\textbf {\bibinfo {volume}
  {2}},\ \bibinfo {pages} {016114} (\bibinfo {year} {2024})}\BibitemShut
  {NoStop}%
\bibitem [{\citenamefont {Stern}\ \emph {et~al.}(2025)\citenamefont {Stern},
  \citenamefont {Guzman}, \citenamefont {Martins}, \citenamefont {Liu},\ and\
  \citenamefont {Balasubramanian}}]{stern2025physical}%
  \BibitemOpen
  \bibfield  {author} {\bibinfo {author} {\bibfnamefont {M.}~\bibnamefont
  {Stern}}, \bibinfo {author} {\bibfnamefont {M.}~\bibnamefont {Guzman}},
  \bibinfo {author} {\bibfnamefont {F.}~\bibnamefont {Martins}}, \bibinfo
  {author} {\bibfnamefont {A.~J.}\ \bibnamefont {Liu}},\ and\ \bibinfo {author}
  {\bibfnamefont {V.}~\bibnamefont {Balasubramanian}},\ }\bibfield  {title}
  {\bibinfo {title} {Physical networks become what they learn},\ }\href
  {https://doi.org/10.1103/PhysRevLett.134.147402} {\bibfield  {journal}
  {\bibinfo  {journal} {Physical Review Letters}\ }\textbf {\bibinfo {volume}
  {134}},\ \bibinfo {pages} {147402} (\bibinfo {year} {2025})}\BibitemShut
  {NoStop}%
\bibitem [{\citenamefont {Guzman}\ \emph {et~al.}(2025)\citenamefont {Guzman},
  \citenamefont {Martins}, \citenamefont {Stern},\ and\ \citenamefont
  {Liu}}]{guzman2025microscopic}%
  \BibitemOpen
  \bibfield  {author} {\bibinfo {author} {\bibfnamefont {M.}~\bibnamefont
  {Guzman}}, \bibinfo {author} {\bibfnamefont {F.}~\bibnamefont {Martins}},
  \bibinfo {author} {\bibfnamefont {M.}~\bibnamefont {Stern}},\ and\ \bibinfo
  {author} {\bibfnamefont {A.~J.}\ \bibnamefont {Liu}},\ }\bibfield  {title}
  {\bibinfo {title} {Microscopic imprints of learned solutions in tunable
  networks},\ }\href {https://doi.org/10.1103/f2hb-c9s1} {\bibfield  {journal}
  {\bibinfo  {journal} {Physical Review X}\ }\textbf {\bibinfo {volume} {15}},\
  \bibinfo {pages} {031056} (\bibinfo {year} {2025})}\BibitemShut {NoStop}%
\bibitem [{\citenamefont {Kunin}\ \emph {et~al.}(2021)\citenamefont {Kunin},
  \citenamefont {Sagastuy-Brena}, \citenamefont {Ganguli}, \citenamefont
  {Yamins},\ and\ \citenamefont {Tanaka}}]{kunin2021neural}%
  \BibitemOpen
  \bibfield  {author} {\bibinfo {author} {\bibfnamefont {D.}~\bibnamefont
  {Kunin}}, \bibinfo {author} {\bibfnamefont {J.}~\bibnamefont
  {Sagastuy-Brena}}, \bibinfo {author} {\bibfnamefont {S.}~\bibnamefont
  {Ganguli}}, \bibinfo {author} {\bibfnamefont {D.~L.~K.}\ \bibnamefont
  {Yamins}},\ and\ \bibinfo {author} {\bibfnamefont {H.}~\bibnamefont
  {Tanaka}},\ }\bibfield  {title} {\bibinfo {title} {Neural mechanics: Symmetry
  and broken conservation laws in deep learning dynamics},\ }in\ \href@noop {}
  {\emph {\bibinfo {booktitle} {International Conference on Learning
  Representations (ICLR)}}}\ (\bibinfo {year} {2021})\BibitemShut {NoStop}%
\bibitem [{\citenamefont {Marcotte}\ \emph {et~al.}(2023)\citenamefont
  {Marcotte}, \citenamefont {Gribonval},\ and\ \citenamefont
  {Peyr{\'e}}}]{marcotte2023abide}%
  \BibitemOpen
  \bibfield  {author} {\bibinfo {author} {\bibfnamefont {S.}~\bibnamefont
  {Marcotte}}, \bibinfo {author} {\bibfnamefont {R.}~\bibnamefont
  {Gribonval}},\ and\ \bibinfo {author} {\bibfnamefont {G.}~\bibnamefont
  {Peyr{\'e}}},\ }\bibfield  {title} {\bibinfo {title} {Abide by the law and
  follow the flow: Conservation laws for gradient flows},\ }in\ \href@noop {}
  {\emph {\bibinfo {booktitle} {Advances in Neural Information Processing
  Systems (NeurIPS)}}}\ (\bibinfo {year} {2023})\BibitemShut {NoStop}%
\end{thebibliography}%

\appendix

\section{Numerical methods}
\label{app:numerics}

\emph{States.} Free, nudged, and reference states are computed from the
KKT systems of the constrained energy minimizations. In linear circuits
these are single bordered linear solves. In diode circuits the co-content
is convex and piecewise quadratic, and states are computed by semismooth
Newton iteration on the equality-constrained KKT system, with tolerance
$10^{-12}$ and warm starts along trajectories. Rectifiers carry a reverse
leakage $\epsilon = 10^{-4}$, which keeps the Hessian nondegenerate on
connected circuits.

\emph{Flows.} The zero-nudge flows are integrated with fourth-order
Runge-Kutta. Step control exploits the theory. For EP and CL the drift of
$K$, exactly conserved in the linear case, serves as the local error
metric. For AL the step is controlled against the exact dissipation rate
$\dot K = -2\Phi^*$ of Proposition~\ref{prop:dissipation}. In nonlinear
circuits, step-doubling error control is used instead. An integration
step that would drive a conductance negative is shortened to the zero
crossing and the edge is removed, following
Ref.~\cite{mcginnis2026conservation}.

\emph{Matched-loss protocol.} Each rule is integrated once per task, with
snapshots taken the first time the training loss crosses each prescribed
level. All comparisons in Sec.~\ref{sec:generalization} are between
snapshots at identical levels. The residual floor quoted there is
$\sqrt{2L}$, the function-space resolution implied by training loss $L$.

\emph{Validation.} The implementation reproduces the closed-form
voltage-divider dynamics of Ref.~\cite{mcginnis2026conservation} to
$10^{-10}$, conserves $K$ along EP and CL to $10^{-12}$ over full
trajectories, and confirms Remark~\ref{rem:equiv} numerically:
single-output EP and CL $\omega$-limits agree to $2\times10^{-9}$ in
relative distance. The nonlinear module reduces exactly to the linear one
when no diodes are present and reproduces the scale invariance of
Proposition~\ref{prop:homogeneity} to $10^{-8}$ in its all-linear control
configuration.

\end{document}